\documentclass[12pt]{article}
\usepackage[centertags]{amsmath}
\usepackage{amsfonts,amsthm,amssymb}
\usepackage{amssymb}
\usepackage{amsmath}
\usepackage{graphicx}
\usepackage{MnSymbol}
\usepackage{color}
\usepackage{tikz}
\usepackage[margin=1.00in]{geometry}
\theoremstyle{definition}
\newtheorem{theorem}{Theorem}
\newtheorem{corollary}{Corollary}
\newtheorem{definition}{Definition}
\newtheorem{lemma}{Lemma}
\newtheorem{proposition}{Proposition}

\newcommand{\rset}{\mathbb{R}}

\newcommand{\val}{\textsf{Val}}
\begin{document}

\title{Random Extensive Form Games and its Application to Bargaining}
\author{Itai Arieli, Yakov Babichenko}

\maketitle

\begin{abstract}
We consider two-player random extensive form games where the payoffs at the leaves are independently drawn uniformly at random from a given feasible set $C$. We study the asymptotic distribution of the subgame perfect equilibrium outcome for binary-trees with increasing depth in various random (or deterministic) assignments of players to nodes. We characterize the assignments under which the asymptotic distribution concentrates around a point. Our analysis provides a natural way to derive from the asymptotic distribution a novel solution concept for two-player bargaining problems with a solid strategic justification. 
\end{abstract}

\section{Introduction}
A subgame perfect equilibrium is one of the fundamental solution concepts in game theory. Its characterization is very simple: A strategy profile is a subgame perfect equilibrium if it constitutes a Nash equilibrium of every subgame of the original game. In extensive form games with perfect information, subgame perfect equilibrium coincides with backward-induction and can be easily determined.     

In this work we study the distribution of the subgame perfect equilibrium outcome in random two-player extensive form games. In these games 
there are two potential natural degrees of randomness that may be considered. The first is random payoffs where the payoffs at the leaves of the tree are randomly drawn.  The second is random assignment of players to nodes, where one of the two decision makers at every node is drawn at random in accordance with a certain distribution. We restrict attention to binary trees with two players and focus on the case where the payoffs are i.i.d.\ uniform random draws over a given domain $C$. Our analysis combines the two potential sources of randomness and relates properties of the subgame perfect equilibrium outcome to a solution concept in the theory of bargaining.

The existing literature on random games focuses on normal form games and studies properties such as the expected number of Nash equilibrium \cite{McLen05,BM05}, the distribution of pure Nash equilibria \cite{Dresher70,Powers90,Stanford95,Papa95,RS00,Taka08} or the maximal number of equilibria \cite{McLen97}. Here we restrict our attention to random extensive form games. We consider binary random games and study the asymptotic properties of the subgame perfect equilibrium outcome when the depth is increasing. The fundamental natural question that arises is under which condition the limit exists and whether in some sense the subgame perfect distribution \emph{concentrates} around a certain point, i.e., whether the subgame perfect equilibrium outcome is close to a certain point with probability close to 1, for sufficiently deep trees. Such a concentration point, if exists, may serve as a benchmark solution for a bargaining problem with $C$ as its feasible set. Indeed this point has quite solid strategic justification: For \emph{almost all}\footnote{By ``almost all" we mean in the context of the standard Lebesgue measure (when we refer to games as points in $\mathbb{R}^k$ for the correct dimension $k$). Note that analyzing probability under i.i.d. uniform distribution is equivalent to analyzing the Lebesgue measure.} strategic interactions of a quite natural type with outcomes in $C$, the solution of the interaction will be close to this concentration point, i.e., for large enough perfect information games over binary trees (with a particular assignment of players to control the nodes) and with payoffs in $C$. Hence we are interested in characterizing the cases where concentration occurs, and these cases will induce interesting incites into bargaining theory.

Our main theorem introduces the \emph{random extensive form} bargaining solution that naturally arises from random extensive form games (see Section \ref{section:atb}). We characterize our solution concept by introducing a sequence of increasing random extensive form games for which concentration around a point holds. This point defines our solution concept and satisfies the three core axioms of a standard solution concept in bargaining: Symmetry, efficiency and scale invariance.

\subsection{Application to Bargaining}\label{sec:in-bar}
Starting with Nash~\cite{Nash50} in 1950, the bargaining problem has been widely studied from different perspectives. Two complementary trusts of the bargaining theory are the \emph{axiomatic approach} and \emph{implementation theory}.
The axiomatic approach aims to understand what should be the solution of a bargaining problem by introducing basic properties that it should satisfy. However, this approach provides no strategic reasoning for how this solution arises. The so-called ``Nash program" aims to support solutions in a non-cooperative framework. This is exactly where implementation theory comes to play, where the goal is to design a game (i.e., a bargaining mechanism) whose equilibrium is the solution.

Many axioms have been suggested in the context of bargaining. Among the suggested axioms, the \emph{standard} axioms that are satisfied by most of the classical solutions (e.g., Nash solution \cite{Nash50}, Kalai-Smorodinski solution \cite{KS75}, and others) are Pareto efficiency, Symmetry, and Invariance with respect to positive affine transformations. The terminology of Kalai-Smorodinsky \cite{KS75} (see also \cite{Troc04}), defines a solution to be \emph{standard} if it satisfies these three axioms. Our focus will be on implementing a standard solution.

Implementation theory in bargaining, and in particular implementation in subgame perfect equilibrium, has been widely studied. Several mechanisms \cite{Moulin84,BRW86,Howard92,Anb93,Miya02} have been suggested for implementing standard bargaining solutions. All these mechanisms are based on the idea that players sequentially offer an allocation that should be approved by the other players. 

We take a different approach to the implementation problem, and we pose the question: Can a random extensive form game serve as an implementation to a standard bargaining solution? Theorem \ref{theo:bar} constitutes a positive answer to this question. This approach differs from the above \cite{Moulin84,BRW86,Howard92,Anb93,Miya02} in the following aspects. 

First, the designed game is \emph{not} based on offers and approvals. Second, it has minimalistic requirements about the abilities of the designer: It is sufficient that the designer will be able to draw outcomes from the bargaining set uniformly at random (as outcomes of the game) in order to design a game whose equilibrium outcome is (close to) a standard solution. Third, our result, in some sense, answers the following type of questions: How hard is it to implement a standard solution? How structured should the designed game be? Our result can be interpreted as follows, \emph{almost all} binary choice large enough extensive form games (with a particular random structure of the decision maker over the nodes) implement a standard bargaining solution.

\subsection{Outline of the Results}\label{sec:out-res}
For ease of clarity of the presentation, throughout the paper we focus on a specific simple setting. As shown in Comment 1 of Section \ref{sec:discussion}, our analysis holds in much more general settings.

We let $C\subset\rset^2$ be a compact convex set\footnote{Convexity of the set $C$ is not essential; see Comment 1 of Section \ref{sec:discussion}.}. 
%We denote by 
%\begin{align*}
%Par(C):=\{c\in C: \text{ There is no } c' \text{ such that } c_1<c'_1 \text{ and } c_2<c'_2 \}.
%\end{align*}
%the \emph{Pareto efficient boundary} of $C$. We focus on sets $C$ such that  $Par(C)$ is also simply connected. 
Let $\mu_0$ be the uniform distribution over\footnote{Our analysis holds for every non-atomic probability distribution $\mu_0$ with positive density over $C$; see Comment 1 of Section \ref{sec:discussion}.} $C$. The game is played over a complete binary tree of depth $n$. The $2^n$ payoffs at the leaves are random independent draws in accordance with $\mu_0$.

We consider two benchmark models of assigning players to control nodes in the tree:
\begin{enumerate}
\item \emph{The alternating game.} Player 1 controls all nodes at depths $n-1,n-3,n-5,...$ Player 2 controls all nodes at depths $n-2,n-4,n-6,...$. Such a game (with random payoffs) will be denoted by $A_n$.
\item \emph{The random controlling player game.} Each Player $i=1,2$ controls each node with probability $\frac{1}{2}$, where the controlling player at each node is drawn independently across the nodes and independently of the payoffs. Such a game (with random controlling players and random payoffs) will be denoted by $R_n$.
\end{enumerate}
In Section \ref{sec:bar} we consider additional assignments of players to nodes, which combines these two models.

For both players $i=1,2$, both random games $A_n$ and $R_n$ have no two equal payoffs for player $i$ with probability one. Therefore, both $A_n$ and $R_n$ have almost surely unique subgame perfect equilibrium. We denote by $\val(A_n)$ and $\val(R_n)$ the random variable of the \emph{games' value}, which is the payoff profile at the unique subgame perfect equilibrium. 
%depends on the players assignment. 
We are interested in the asymptotic properties of the values $\val(A_n)$ and $\val(R_n)$ when $n\rightarrow \infty$. In particular are interested in the phenomena of \emph{concentration of the value} which holds if there exists a point $c\in C$ such that the probability that $\val(T_n)$ lies in a given neighborhood of
$c$ goes to one as the depth of the tree goes to infinity. Formally, 
\begin{definition}
For a sequence of random games $(T_n)_{n=1}^\infty$, we say that $\val(T_n)$ \emph{concentrates around the point} $c\in R^2$ if the sequence $\val(T_n)$ weakly converges to the Dirac measure on the singleton $c$, i.e., $\val(T_n) \Rightarrow \delta_c$.
%Equivalently, for every $\varepsilon>0$ there exists $n_0$ such that for every $n>n_0$ holds
%\begin{align}\label{eq:conv-def}
%Pr(||\val(T_n)-c||_\infty>\varepsilon)<\varepsilon.
%\end{align}
\end{definition}

%In the spirit of Rubinstein alternating offers model for Bargaining we consider first the alternating model where player $1$ controls all nodes of depth $n$, player $2$ controls all nodes of depth $n-1$ and so forth.   

We start our analysis with the zero-sum case, namely $C=\{(x,-x)|0\leq x\leq 1\}$. The results for random zero-sum games are as follows:
\begin{itemize}
\item For the alternating game, $\val(A_n)$ concentrates around the point $(b,-b)$, where $b=\frac{\sqrt{5}-1}{2}$ (Theorem \ref{th:zs}).
\item For the random controlling player model, however, we do not have a concentration result. In fact, we have $\val(R_n)=\mu_0$, i.e., the value is distributed according to the distribution of the initial payoffs for all $n$ (Proposition \ref{pro:zs}).  
\end{itemize}

Then we turn to the non-zero sum case where $C$ is a convex set with non-empty interior. The results for this case are as follows:
\begin{itemize}
\item For the alternating game, $\val(A_n)$ concentrates around a Pareto efficient point (Theorem \ref{th:gd}).
\item For the random controlling player game, although we do not have a concentration around a point, we do have a concentration of $\val(R_n)$ around the Pareto efficient boundary (Proposition \ref{pro:gen}).
\end{itemize}

Our next target is to examine whether random extensive form games can serve as an implementation to a bargaining solution. 
The set $C$ can be considered as the feasible set of alternatives. A solution concept assigns a unique point for every such feasible set. Therefore, the first property that should be satisfied for such a random game is
\begin{enumerate}
\item Concentration of the value around a point.
\end{enumerate}
Another two natural properties that are common to most solution concepts in bargaining (the \emph{standard axioms}) are:
\begin{enumerate}
\setcounter{enumi}{1}
\item Pareto efficiency: The unique solution should be Pareto efficient.
\item Symmetry: For every symmetric sets $C$ the solution should be symmetric. 
\end{enumerate}

By considering the results mentioned so far, we see that none of the suggested random extensive form games enjoys all three properties. The alternating game concentrates around a Pareto efficient point (Theorem \ref{th:gd}). However, the concentration point does not satisfy symmetry (Theorem \ref{th:zs}); Player 1 who acts last has an advantage. The random controlling player game is symmetric by nature; however it fails to satisfy concentration (Proposition \ref{pro:zs}).

It turns out that a random extensive form game, which satisfies all the above desirable properties can be obtained by a ``hybrid" of the two discussed cases (the alternating and the random). In the \emph{hybrid game} all nodes at depth $n-k$ for odd $k$ are controlled by Player $1$ with probability $\frac{1}{2}+\varepsilon$ and by Player $2$ with the remaining probability $\frac{1}{2}-\varepsilon$. All nodes at depth $n-k$ for even $k$ are controlled by Player $2$ with probability $\frac{1}{2}+\varepsilon$ and by Player $1$ with probability $\frac{1}{2}-\varepsilon$. Such a game (with random controlling players and random payoffs) will be denoted by $H^\varepsilon_n$. Note that $H^{1/2}_n=A_n$ and $H^0_n=R_n$.
We show in Theorem \ref{theo:bar} that for small values of $\varepsilon$ the hybrid game $H^\varepsilon_n$, indeed (approximately) satisfies the above-mentioned desirable property:
\begin{enumerate}
\item For every fixed $\varepsilon>0$, the value of $H^\varepsilon_n$ concentrates around a point.
\item For every fixed $\varepsilon>0$, the concentration point of $\val(H^\varepsilon_n)$ is Pareto efficient.
\item For a symmetric set $C$, given $\delta>0$ we can set $\varepsilon>0$ such that the concentration point of $\val(H^\varepsilon_n)$ will be $\delta$-close to the diagonal (i.e., $\delta$-close to the symmetric solution).
\end{enumerate}

\section{Zero sum games}

We first consider the case of random zero-sum games, where $C=conv\{(0,0),(1,-1)\}$, i.e., Player 1's payoff is uniformly distributed in $[0,1]$. For convenience of notations, along this section we drop the payoff of player 2, which is equal to minus player's 1 payoff. 

\subsection{Alternating game}
The following theorem states that the value of the alternating zero-sum game converges to the golden ratio minus one.

\begin{theorem}\label{th:zs}
$\val(A_n)$ concentrates around point $b$ where $b=\frac{\sqrt{5}-1}{2}\approx 0.62$.
\end{theorem}
\begin{proof}
We denote $\mu_n:=\val(A_n)$. By the backward induction procedure we can deduce a recursive formula for the sequence of measures $(\mu_n)_{n=1}^\infty$. The value of a game of depth $n+1$ is a function of the values of the two subgames of depth $n$ (which correspond to the two possible actions of the player that controls the root). Note that the two values of the two subgames are two independent draws according to $\mu_{n}$. We denote by $\min(\nu,\nu)$ and $\max(\nu,\nu)$ the minimum and the maximum of two independent draws according to $\nu$, and we have the following recursive formula
\begin{align*}
\mu_{n+1}=
\begin{cases}
\max (\mu_{n},\mu_{n}) & \text{ if } n \text{ is even} \\
\min (\mu_{n},\mu_{n}) & \text{ if } n \text{ is odd}
\end{cases}
\end{align*}
because Player 1 will choose the maximal value (among the two), and Player 2 will choose the minimal.
%The payoffs at depth $n$ are also the values of a game of depth $0$ and are distributed according to the uniform distribution; i.e., $\mu_0=Un[0,1]$. The payoffs at depth $n-1$ are also the values of a game of depth $1$ and are distributed according to the distribution $\mu_1=\max(\mu_0,\mu_0)$ because player 1 acts at depth $n-1$. The payoffs at depth $n-2$ are also the values of a game of depth $2$ and are distributed according to the distribution $\mu_2=\min(\mu_1,\mu_1)$ because player 2 acts at depth $n-1$. Using this argument we deduce that
%\begin{align*}
%\mu_{k+1}=
%\begin{cases}
%\max (\mu_{k},\mu_{k}) & \text{ if } k \text{ is even} \\
%\min (\mu_{k},\mu_{k}) & \text{ if } k \text{ is odd}
%\end{cases}
%\end{align*}

%Given $\varepsilon>0$, by definition of concentration around the point $b$ (see equation \eqref{eq:conv-def}), it is sufficient to prove that for large enough $n$ holds $\mu_n([0,b-\varepsilon])< \frac{\varepsilon}{2}$ and $\mu_n([b+\varepsilon,1])< \frac{\varepsilon}{2}$. First we bound the expressions $\mu_n([0,b-\varepsilon])$ and $\mu_n([b+\varepsilon,1])$ for even values of $n$.

For $x\in [0,1]$ we denote $F_n(x)=\mu_n([0,x])$. For $n=2k$ we have 
\begin{align}\label{eq:even-iter}
F_{2k+1}(x)=\left( F_{2k}(x)\right)^2
\end{align}
because the maximum of two independent random variables is below $x$ iff both of them are below $x$. Note also that
\begin{align}\label{eq:odd-iter}
F_{2k+2}(x)=2F_{2k+1}(x)-\left(F_{2k+1}(x) \right)^2=2\left( F_{2k}(x)\right)^2 - \left( F_{2k}(x)\right)^4.
\end{align}
The first equality follows from the inclusion-exclusion principle. The second inequality follows from equation \eqref{eq:even-iter}.
Let $\varphi(x)=2x^2-x^4$, and let $\varphi^k$ be $k$ times composition of $\varphi$ on itself, i.e.,
$$\varphi^k(x)=\underbrace{\varphi\circ\cdots\circ \varphi}_\text{$k$ times}(x).$$
It readily follows by induction from equation \eqref{eq:odd-iter} that
\begin{align}\label{eq:odd-iter2}
F_{2k}(x)=\varphi^k(F_0(x))=\varphi^k(x).
\end{align}
 
Note that for $0\leq x \leq 1$, $\varphi(x)$ is a monotonically increasing polynomial with three fixed points: $0,b,1$. In addition, $\varphi(x)<x$ for $0<x<b$, and $\varphi(x)>x$ for $b<x<1$ (see Figure \ref{fig:2x2-x4}).

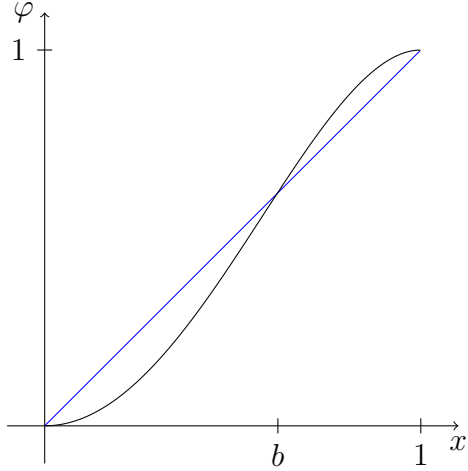
\begin{figure}[h]
\centering
\caption{The function $\varphi(x)=2x^2-x^4$.}\label{fig:2x2-x4}
\begin{tikzpicture}[xscale=5,yscale=5,domain=0:1,samples=100]\label{fig:2x2-x4}
    \draw[->] (-0.1,0) -- (1.1,0) node[below] {$x$};
    \draw[->] (0,-0.1) -- (0,1.1) node[left] {$\varphi$};
    \draw (1,0.02) -- (1,-0.02) node[below] {$1$};
    \draw (0.62,0.02) -- (0.62,-0.02) node[below] {$b$};
    \draw (0.02,1) -- (-0.02,1) node[left] {$1$};
    \draw[blue] plot (\x,{\x});
    \draw[black] plot (\x,{2*\x*\x-\x*\x*\x*\x});
\end{tikzpicture}
\end{figure}

For every $x<b$ we consider the sequence $(F_{2k}(x))_{k=0}^\infty=(\varphi^k(x))_{k=0}^\infty$, which is monotonically decreasing because 
\begin{align*}
\varphi^{k+1}(x)=\varphi^k(\varphi(x))\leq \varphi^k(x).
\end{align*}
Let $\lim_{k\rightarrow \infty} \varphi^k(x)=l(x)$. The sequence is monotonically decreasing; therefore, $l(x)$ is well defined and $l(x)\leq x <b$. Note also that $l(x)$ is a fixed point of $\varphi$ because
\begin{align*}
l(x)=\lim_{k\rightarrow \infty} \varphi^k(x)=\lim_{k\rightarrow \infty} \varphi(\varphi^k(x))=\varphi(l(x)).
\end{align*}
Hence $l(x)=0$.

Using similar arguments we can show that for $x>b$ the sequence $\lim_{k\rightarrow \infty} \varphi^k(x)$ is monotonically increasing and its limit is the fixed point 1. Summarizing, we have
\begin{eqnarray}
\lim_{k\rightarrow\infty} F_{2k}(x)=\lim_{k\rightarrow\infty}\varphi^k(F_0(x))=\begin{cases}
0\ \ \text{ if } 0<x<b \\ 
b\ \ \text{ if } x=b.\\
1\ \ \text{ if } b<x\leq 1.
\end{cases}%  
\end{eqnarray}
This implies that $F_{2k}$ converges to a limit Dirac measure concentrated in $b$. The fact that the sequence $\{F_{2k+1}\}_k$ converges to the same limit readily follows from equation \eqref{eq:even-iter}. 

We have proved that the CDF of $\mu_n$ pointwise converges to the CDF of the Dirac measure, which implies weak convergence.
%the fact that if $F_{2k}$ is close to the cdf of the Dirac measure on $b$ then by equation \eqref{eq:even-iter} $F_{2k+1}$ remains close to the cdf of the Dirac measure. 
%Therefore, the fact that  $\{F_{2k}\}_k$ converges to a Dirac measure implies that the sequence $\{F_{2k+1}\}_k$ converges to the same limit.
\end{proof}

\subsection{Random controlling player game}

The following observation demonstrates that unlike alternating random games the value of a random zero-sum game with random structure does not concentrate.

\begin{proposition}\label{pro:zs}
For every $n$, $\val(R_n)$ is the uniform distribution over $[0,1]$.
\end{proposition}

\begin{proof}
Using similar arguments to those in the proof of Theorem \ref{th:zs}, we can write a recursive formula for the sequence of probability measures $(\nu_n)_{n=1}^\infty = (\val(R_n))_{n=1}^\infty$. We argue that 
\begin{align}\label{eq:rec-sym-zs}
\nu_{k+1}=\frac{1}{2}\max(\nu_k,\nu_k)+\frac{1}{2}\min(\nu_k,\nu_k).
\end{align}
This follows from the fact that with probability $\frac{1}{2}$ Player 1 controls a certain vertex at depth $n-k-1$, and then he chooses the maximal payoff. With probability $\frac{1}{2}$ Player 2 controls this vertex, and then he chooses the minimal payoff. From this recursive formula we deduce that
\begin{align*}
\nu_{k+1}([0,x])=\frac{1}{2}(\nu_k([0,x]))^2+\frac{1}{2}[2\nu_k([0,x])-(\nu_k([0,x]))^2]=\nu_k([0,x]).
\end{align*}
Namely, the sequence of measures $(\nu_n)_{n=1}^\infty$ remains constant for every initial probability measure $\nu_0$.
\end{proof}

\section{General Domains}
Now we consider the case where $C\subset \rset^2$ is a general compact convex set with non-empty interior.

\begin{definition}\label{def:c}
For a point $x\in \mathbb{R}^2$ we denote $C_{>>}(x):=\{c\in C:c_1>x_1 \text{ and } c_2>x_2\}$ the first quadrant with the origin at point $x$. Similarly we denote $C_{<>}(x)$, $C_{<<}(x)$, and $C_{><}(x)$ the second, third, and fourth quadrants respectively, see Figure \ref{pic:defC}. More generally we denote $C_{\square \blacksquare}(x):=\{c\in C:c_1 \square x_1 \text{ and } c_2 \blacksquare x_2\}$ where $\square, \blacksquare \in \{\leq, <, \geq, > \}$.
\end{definition}

\begin{figure}[h]
\begin{center}
\includegraphics[scale=1]{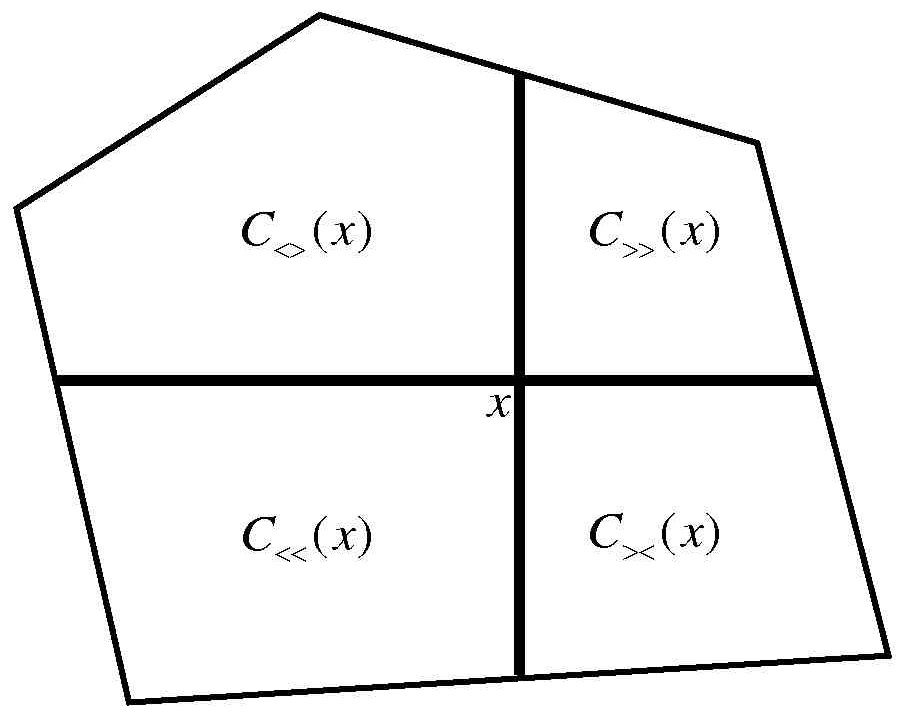}
\end{center}
\caption{The sets $C_{>>}(x)$, $C_{<>}(x)$, $C_{<<}(x)$, and $C_{><}(x)$.}\label{pic:defC}
\end{figure}
 
\subsection{Alternating game}
As in the zero-sum case, the value of the alternating game concentrates around a point. Moreover the concentration point is Pareto efficient.
\begin{theorem}\label{th:gd}
The value of the alternating game $\val(A_n)$ concentrates around a Pareto efficient point.
\end{theorem}
\begin{proof}
For $x\in \mathbb{R}$ we denote $C_1(x)=\{(x',y')\in C:x'\leq x\}$, similarly, for $y\in \mathbb{R}$ we denote $C_2(y)=\{(x',y')\in C:y'\leq y\}$. We let $F^1_n(x):=\nu_n(C_1(x))$ be the CDF of the marginal of $\nu_n$ to Player 1's payoffs and $F^2_n(y)=\nu_n(C_2(y))$ to Player 2's payoffs. Let $x_n$ be the number $x$ for which $F^1_n(x)=b$ and similarly $y_n$ is defined by the number $y$ for which $F^2_n(y)=b$. 
\begin{lemma}\label{mon:lemma}
Both $(x_{2n})_{n=1}^\infty$ and $(y_{2n+1})_{n=1}^\infty$ are increasing sequences. 
\end{lemma}
\begin{proof}
%We let $C^{x}\subset C$ be the subset of all points $(z,w)\in C$ such that $z\leq x$. 
In order to show that $x_{2n}$ is increasing with $n$ it is sufficient to show that 
\begin{equation}\label{eq:smb}
\mu_{2n+2}(C_1(x_{2n}))\leq b.
\end{equation}
Since if \eqref{eq:smb} holds, then $F_{2n+2}(x_{2n})\leq b$, and by the monotonicity of the CDF we have that $x_{2n+2}\geq x_{2n}$. 

First, note that $\mu_{2n}(C_1(x_{2n}))=b$ 
and hence $\mu_{2n+1}(C_1(x_{2n}))=b^2$; this follows directly by a similar consideration as in equation \eqref{eq:even-iter}.

Now we consider step $2n+2$. Clearly, $\mu_{2n+2}(C_1(x_{2n}))$ is bounded from above by the probability that two i.i.d.\ draws that are distributed according to $\mu_{2n+1}$ lies in $C_1(x_{2n})$ (because $\mu_{2n+2}$ is essentially obtained by a selection rule among i.i.d.\ outcomes that are distributed according to $\mu_{2n+1}$). Therefore we have 
\begin{equation}\label{eq:smb2}
\mu_{2n+2}(C_1(x_{2n}))\leq 1-(1-b^2)^2=2b^2-b^4=\varphi(b)=b.   
\end{equation}
%where $1-(1-b^2)^2$ is exactly the probability of appearance of an outcome in $C(x_{2n})$ among two outcomes that are distributed according to $\mu_{2n+1}$.

A symmetric consideration applied to Player $2$ shows that $y_{2n+1}$ is increasing.
\end{proof}    
%In order to prove that $\val(T_n)$ converges to a point mass concentrated on the pareto efficient frontier  we will show first that $(x_{2n},y_{2n+1})$ converges  to a  pareto efficient point in $C$.
%\begin{lemma}
%The limit, $\lim_{n\rightarrow\infty}(x_{2n},y_{2n+1})=(x^*,y^*)$ lies in the Pareto efficient boundary of $C$.
%\end{lemma} 
%\begin{proof}

Let $x^*=\lim_{n\rightarrow \infty} x_{2n}$, and $y^*=\lim_{n\rightarrow \infty} y_{2n+1}$.
We first contend that if $x<x^*$ then $\lim_{n}\mu_{2n}(C_1(x))=0$.
To see this note that, by definition $\mu_{2k}(C_1(x))<b$ for some $k$. A similar derivation to those applied in the proof of Theorem \ref{th:zs} yields that 
\begin{align}\label{eq:av1}
\lim_{n}\mu_{2n}(C_1(x))\leq \lim_{n}\varphi^{n}(\mu_{2n}(C_1(x)))=0,
\end{align}
where the last equality follows from the fact that $\mu_{2k}(C^x)<b$. From this we can deduce that $\lim_{n}\mu_{n}(C_1(x))=0$, because for odd steps we can bound $\mu_{2n+1}(C_1(x))\leq 2\mu_{2n}(C_1(x))$.

Similar arguments show that for every $y<y^*$,
\begin{align}\label{eq:av2}
\lim_{n}\mu_{n}(C_2(y))=0.
\end{align} 

%This in particular implies that for every $x<x^*$ and every $y<y^*$,
%\begin{equation}\label{eq:av}
%\lim_{n}\mu_{n}(C_y)=\mu_{n}(C^x)=0.
%\end{equation} 

We next show that $(x^*,y^*)$ lies on the Pareto efficient boundary.   
%For every $(x,y)$ let 
%\begin{align*}
%C_{\geq \geq}(x,y):=\{(x',y')\in C : x'\geq x \text{ and } y'\geq y \}.
%\end{align*}
Equations \eqref{eq:av1} and \eqref{eq:av2} imply that $C_{\geq \geq}(x,y)$ is nonempty for every $x<x^*,$ and $y<y^*$ (see Definition \ref{def:c}), and therefore $C_{\geq \geq }(x^*,y^*)$ is nonempty.

Assume by way of contradiction that $(x^*,y^*)$ is Pareto dominated. Let $(x',y')$ be any point such that $x'>x^*,$ $y'>y^*,$ and  $C_{\geq \geq}(x',y')$ has a nonempty interior. In particular $\mu_0(C_{\geq \geq}(x',y'))>0$.
We claim first that for every $n$,
\begin{equation}\label{eq:rec2}
\mu_{2n+1}(C_{\geq \geq}(x',y'))\geq 2b\mu_{2n}(C_{\geq \geq}(x',y')).
\end{equation}
Because with probability at least $b\mu_{2n}(C_{\geq \geq}(x',y'))$ the value of the left subtree is in $C_1(x*)$ and the value of the right subtree is in $C_{\geq \geq}(x',y')$; in such a case Player 1 will choose the outcome from $C_{\geq \geq}(x',y')$. However, we also have the disjoint event where the left and right realizations are flipped (this event is disjoint because $C_1(x*)\cap C_{\geq \geq}(x',y')=\emptyset$) where again the chosen outcome is in $C_{\geq \geq}(x',y')$. Therefore $2b\mu_{2n}(C_{\geq \geq}(x',y'))$ is a lower bound on $\mu_{2n+1}(C_{\geq \geq}(x',y'))$.
%To see this note that $\mu_{2n+1}(C^{(x',y')})$ represents the probability that an outcome from $C^{(x',y')}$ will be the realised value of the random game $T_{2n+1}$. By backward induction, since player $1$ controls the root of $T_{2n+1}$ this probability equals the probability that either the two outcomes of the two sub-trees of the root are in $(C^{(x',y')})$ which is exactly $\mu_{2n}(C^{(x',y')})^2,$ or, only one of the backward induction outcomes lies in $C^{(x',y')}$ and the other outcome is less than $x'$, which has a probability of $2\mu_{2n}(C^{(x',y')})F^1_{2n}(x').$      
%Adding these two expressions yields equation \eqref{eq:rec}. Since, $x'>x_{2n}$ we have by definition that $F^1_{2n}(x')>b$. Hence equation \eqref{eq:rec} implies that for every $n$,
%\begin{equation}\label{eq:rec2}
%\mu_{2n+1}(C^{(x',y')})\geq \mu_{2n}(C_{(x',y')})(2b).
%\end{equation}

A similar consideration applied to Player $2$ shows that for every $n$,
\begin{equation}\label{eq:rec3}
\mu_{2n+2}(C_{\geq \geq}(x',y'))\geq 2b \mu_{2n+1}(C_{\geq \geq}(x',y')).
\end{equation}
Together Equations \eqref{eq:rec2} and \eqref{eq:rec3} imply that for every $n$,
$$\mu_{n}(C_{\geq \geq}(x',y'))\geq (2b)^n \mu(C_{\geq \geq}(x',y')).$$
Since $\mu_0(C_{\geq \geq}(x',y'))>0$ this stands in contradiction to the fact that $\mu_{n}(C_{\geq \geq }(x',y'))\leq 1$ for every $n$. We must therefore have that $(x^*,y^*)$ lies in the Pareto efficient boundary.

By Equations \eqref{eq:av1} and \eqref{eq:av2} for every $x'<x^*,$ and $y'<y^*$ we have that, $\lim_n \mu_{n}(C_{\geq \geq}(x',y'))=1$. Since every open set $U$ in $C$ that contains $(x^*,y^*)$ must also contain $C^{(x',y')}$ for some $x'<x^*$ and $y'<y^*$, we must have that $\lim_n \mu_{n}(U)=1$. This shows that $\mu_n$ weakly converges to the Dirac measure $\delta_{(x^*,y^*)},$ and establishes the proof of the theorem.   
\end{proof}

\subsection{Random controlling player game}\label{sec:sym}
%We will now consider games with random structure where the outcomes of a binary tree $R_n$ of depth $n$ are drawn uniformly at random from a convex domain $C$. In addition, each of the two players is assigned randomly to control every given node with probability half. We let $\val(R_n)$ be the random variable that represents the value of this games and $\nu_n$ be the law of $\val(R_n)$.

We denote by $P_\varepsilon(C)$ the set of points within a distance at most $\varepsilon$ from the Pareto efficient boundary. We say that a sequence of measures $\nu_n$ \emph{concentrates on the Pareto efficient boundary} if $\lim_{n}\nu_n(P_\varepsilon(C))=1$ for every $\varepsilon>0$.  
%We first show that $\val(R_n)$ satisfies efficiency.

The following proposition shows that although $\val(R_n)$ does not concentrate around a point, the value is (approximately) Pareto efficient with probability close to 1.
\begin{proposition}\label{pro:gen}
The value of the random controlling player game concentrates on the Pareto efficient boundary.
\end{proposition}
\begin{proof}
Similar to the zero-sum case (equation \ref{eq:rec-sym-zs}), we can deduce the following recursive formula on $(\nu_n)_n$.
\begin{align}\label{eq:rec-sym}
\nu_{n+1}=\frac{1}{2}\max_1(\nu_n,\nu_n) + \frac{1}{2}\max_2(\nu_n,\nu_n)
\end{align}
where $\max_i(\nu_n,\nu_n)$ denotes the maximum over the $i$th coordinate of two independent realizations of $\nu_n$.

We let $O_{\frac{\varepsilon}{2}}(C)$ be the set of points with a distance of precisely $\frac{\varepsilon}{2}$ from the efficient boundary. We will show first that for every $(x,y)\in O_{\frac{\varepsilon}{2}}(C)$ it holds that $\lim_{n}\nu_n(C_{\leq \leq}(x,y))=0$ (see Definition \ref{def:c}).

If not, there exists some $(x',y')\in O_{\frac{\varepsilon}{2}}(C)$ and $\delta>0$ for which $\nu_n(C_{\leq \leq}(x',y'))\geq\delta$ for infinitely many $n$'s.
As before $F^1_n(x)$ represents the probability under $\nu_n$ that the realised outcome for Player $1$ would be less than $x$ and similarly $F^2_n(y)$ for Player $2$.

We bound from below the probability $\nu_{n+1}(C_{>>}(x',y'))$ using the recursive formula \eqref{eq:rec-sym}. There are three events under which the realization of $\nu_{n+1}$ lies in $C_{>>}(x',y')$.
\begin{enumerate}
\item Both realizations of $\nu_{n}$ lie in $C_{>>}(x',y')$. The probability of this event is $(\nu_{n}(C_{>>}(x',y')))^2$.
\item Player 1 controls the root---one realization of $\nu_n$ is in $C_{>>}(x',y')$, and the other realization is in $F^1_{n}(x')$. The probability of the first of the two is clearly $\frac{1}{2}.$ The probability of the second is $2\nu_{n}(C_{>>}(x',y'))F^1_{n}(x')$ since the realizations can be flipped (each can appear in the left subgame or the right subgame). Overall this event holds with probability $\nu_{n}(C_{>>}(x',y'))F^1_{n}(x')$.
\item Player 2 controls the root---one realization of $\nu_n$ lies in $C_{>>}(x',y')$, and the other realization lies in $F^2_{n}(x')$. The probability of this event is $\nu_{n}(C_{>>}(x',y'))F^2_{n}(x')$.
\end{enumerate}
These three events are disjoint (note that (2) and (3) are disjoint because the controlling player defers). Therefore, we deduce that for every $n$
\begin{eqnarray}\label{eq:im1}
 & &\nu_{n+1}(C_{>>}(x',y'))\geq\\
\notag& &\nu_{n}(C_{>>}(x',y'))(\nu_{n}(C_{>>}(x',y'))+F^1_{n}(x')+F^2_n(y')).
\end{eqnarray} 
We can rewrite Equation \eqref{eq:im1} as follows
\begin{eqnarray}\label{eq:im}
 & &\nu_{n+1}(C_{>>}(x',y'))\geq\\
\notag& &\nu_{n}(C_{>>}(x',y'))(\nu_{n}(C_{>>}(x',y'))+\nu_{n}(C_{\leq >}(x',y'))+\nu_{n}(C_{>\leq}(x',y'))+2\nu_n(C_{\leq \leq}(x',y'))).
\end{eqnarray} 
Since $\nu_{n}(C_{>>}(x',y'))+\nu_{n}(C_{<>}(x',y'))+\nu_{n}(C_{><}(x',y'))+\nu_n(C_{\leq \leq}(x',y'))=1$ we conclude that $\nu_{n}(C_{>>}(x',y'))$ increases with $n$. Also if $\nu_n(C_{\leq \leq}(x',y'))\geq\delta$ then 
$$\nu_{n+1}(C_{>>}(x',y'))\geq(1+\delta)\nu_{n}(C_{>>}(x',y')).$$
Hence since  $\nu_n(C_{\leq \leq}(x',y'))\geq\delta$ for infinitely many $n$ we have that for every $k$ there exists $n_0$ such that for every $n>n_0$,
$$\nu_n(C_{>>}(x',y'))\geq(1+\delta)^k\nu_0(C_{> >}(x',y')).$$
This stands in contradiction to the fact that $\nu_0((C_{>>}(x',y')))>0$. 

Fix $\eta>0$. We show that there exists $n_0$ such that for every $n>n_0$,
$$\nu_n(P_{\varepsilon}(C))\geq 1-\eta.$$
Choose $\{(x_i,y_i)\}_{i=1}^m\subset O_{\frac{\varepsilon}{2}}(C)$ to be a finite sequence of points, dense enough such that if $(x',y')\in C$ satisfies 
$$(x',y')\not\in \bigcup_{i=1}^m C_{\leq \leq}(x_i,y_i)$$ 
then $(x',y')\in P_\varepsilon(C)$, i.e., $C\setminus \big(\bigcup_{i=1}^m C_{\leq \leq}(x_i,y_i)\big)\subseteq P_{\varepsilon}(C)$.
By the above we can find large enough $n_0$ such that for every $n>n_0$, and every $1\leq i\leq m$,
$$\nu_n(C_{\leq \leq}(x_i,y_i))\leq\frac{\eta}{m}.$$

We therefore conclude that for every $n>n_0$,
$$\nu_n(\bigcup_{i=1}^m C_{\leq \leq}(x_i,y_i))\leq \eta.$$
Since $C\setminus \big(\bigcup_{i=1}^m C_{\leq \leq}(x_i,y_i)\big)\subseteq P_{\varepsilon}(C)$ we conclude that
$$\nu_n(P_{\varepsilon}(C))\geq 1-\eta.$$
This completes the proof of the proposition.
\end{proof}  

\section{The Hybrid Model}\label{sec:bar}

%We will now consider a random game that combine the sequential model and the simultaneous symmetric model.
%For every $0<\varepsilon<\frac{1}{2}$ let $R^\varepsilon_n$ be the random game that is obtained for a binary tree of depth $n$ as follows: the payoffs at the leafs of the tree are drawn uniformly at random from the convex set $C$. Players at the nodes are drawn at random such that a node at depth $n-k$ is controlled by player $1$ with probability $\frac{1}{2}+\varepsilon$ and by player $2$ with probability $\frac{1}{2}-\varepsilon$, if $k$ is even, and vice versa for odd $k$.

As we did previously, we consider as a warm-up the zero-sum case where $C=[0,1]$. We establish a generalization of Theorem \ref{th:zs} which states that the value of the random hybrid game $H^\varepsilon_n$ converges to $b^\varepsilon=\frac{\sqrt{1+\varepsilon^2}-1+\varepsilon}{2\varepsilon}$. This is indeed a generalization of Theorem $1$ because if we set $\varepsilon=\frac{1}{2}$ then the random game $H^{1/2}_n=A_n$ and $b^\varepsilon=b=\frac{\sqrt{5}-1}{2}$.

\begin{proposition}
$\val(H^\varepsilon_n)$ concentrates around the point $b^\varepsilon=\frac{\sqrt{1+\varepsilon^2}-1+\varepsilon}{2\varepsilon}$.
\end{proposition}
The proof uses similar ideas to the proof of Theorem \ref{th:zs}.
\begin{proof}
For every $n$, let $F_n(x)$ be the CDF of $\val(H^\varepsilon_n)$. For every $n$ we have,
\begin{equation}\label{zs:ev}
F_{2n+1}(x)=(\frac{1}{2}+\varepsilon)(F_{2n}(x))^2+(\frac{1}{2}-\varepsilon)[1-(1-F_{2n}(x))^2],
\end{equation}
and
\begin{equation}\label{zs:od}
F_{2n}(x)=(\frac{1}{2}-\varepsilon)(F_{2n}(x))^2+(\frac{1}{2}+\varepsilon)[1-(1-F_{2n}(x))^2].
\end{equation}
Using equations \eqref{zs:ev} and \eqref{zs:od} we get,
\begin{eqnarray}
\label{eq:sh}& &F_{2n+2}(x)=\\
\notag& &F_{2n}(x)+(2\varepsilon)^2[-2(F_{2n}(x))^3+3(F_{2n}(x))^2-F_{2n}(x)]+(2\varepsilon)^3[-(F_{2n}(x)-(F_{2n}(x))^2)^2].
\end{eqnarray}
Let 
$$\varphi^\varepsilon(x)=x+(2\varepsilon)^2[-2x^3+3x^2-x]+(2\varepsilon)^3[-(x-x^2)^2].$$

Simple algebraic computations show that the function $\varphi^\varepsilon$ has four fixed points
\begin{align*}
x_1=0, x_2=1, x_3 = \frac{\sqrt{1+\varepsilon^2}-1+\varepsilon}{2\varepsilon}=\frac{1}{\sqrt{1+\varepsilon^2}+1-\varepsilon}, x_4=-\frac{\sqrt{1+\varepsilon^2}+1-\varepsilon}{2\varepsilon}.
\end{align*}
Among these fixed points only $x_1,x_2$ and $x_3$ are in the segment $[0,1]$. We denote $b^\varepsilon = x_3(\varepsilon)$, and we note that $b^\varepsilon>\frac{1}{2}$ and $\lim_{\varepsilon\rightarrow 0} b^\varepsilon=\frac{1}{2}$. Note also that $\varphi^\varepsilon(x)<x$ for $0<x<b^\varepsilon$ (for instance because $\frac{d \varphi^\varepsilon}{dx}(0) = 1-4\varepsilon^2<1$) 
%$\phi^\varepsilon(\frac{1}{2})=\frac{1}{2}-\frac{\varepsilon^3}{2}<\frac{1}{2}$) 
and $\varphi^\varepsilon(x)>x$ for $b^\varepsilon<x<1$ (for instance because $\frac{d \varphi^\varepsilon}{dx}(1) = 1-4\varepsilon^2<1$). Figure \ref{fig:phi-eps} illustrates the function $\varphi^\varepsilon$ for $\varepsilon=\frac{1}{4}$.

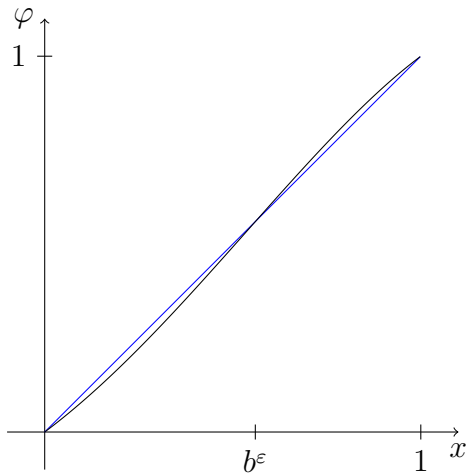
\begin{figure}[h]
\centering
\caption{The function $\varphi^\varepsilon(x)$ for $\varepsilon=\frac{1}{4}$.}\label{fig:phi-eps}
\begin{tikzpicture}[xscale=5,yscale=5,domain=0:1,samples=100]\label{fig:2x2-x4}
    \draw[->] (-0.1,0) -- (1.1,0) node[below] {$x$};
    \draw[->] (0,-0.1) -- (0,1.1) node[left] {$\varphi$};
    \draw (1,0.02) -- (1,-0.02) node[below] {$1$};
    \draw (0.56,0.02) -- (0.56,-0.02) node[below] {$b^\varepsilon$};
    \draw (0.02,1) -- (-0.02,1) node[left] {$1$};
    \draw[blue] plot (\x,{\x});
    \draw[black] plot (\x,{-0.125*\x*\x*\x*\x -0.25*\x*\x*\x +0.625*\x*\x + 0.75*\x});
\end{tikzpicture}
\end{figure}

One can see by equation \eqref{eq:sh} that  
$$F_{2n+2}(x)=\varphi^\varepsilon(F_{2n}(x)).$$
Similar arguments to those applied in Theorem \ref{th:zs} imply that $\val_n(H^\varepsilon_n)$ weakly converges to $b_\varepsilon$ as $n$ grows.
\end{proof}

%The above function $\phi^\varepsilon$ has three fixed points in $[0,1]$: $0,1$ and an intermediate point $b^\varepsilon>\frac{1}{2}$ that approaches $\frac{1}{2}$ when $\varepsilon$ approaches $0$. Moreover $\phi^\varepsilon$ is an $S$ shaped curve that satisfy $0<\phi^\varepsilon(x)<x$ for every $x\in(0,b^\varepsilon)$ and, $1>\phi^\varepsilon(x)>x$ for every $x\in(b^\varepsilon,1)$. For the zero-sum case where $C=[0,1]$, a similar characterization to the one applied in Theorem \ref{th:zs} implies that $\val_n(R^\varepsilon_n)$ weakly converges to the point $b_\varepsilon$ as $n$ grows. 

The following theorem studies the asymptotic properties of $\val_n(H^\varepsilon_n)$ for the general convex domain $C$.
\begin{theorem}\label{theo:bar}
For every $\varepsilon$ the value of the hybrid game concentrates around a Pareto efficient point $(x^\varepsilon,y^\varepsilon)$. Moreover, if the domain $C$ is symmetric then $\lim_{\varepsilon\rightarrow 0} x^\varepsilon - y^\varepsilon =0$; namely, the concentration point approaches the diagonal when $\varepsilon$ goes to $0$.  
\end{theorem}

Before establishing the proof of the Theorem (Section \ref{sec:bar-pr}) we present more formally the applications of this result to bargaining and implementation.

\subsection{Application to Bargaining}\label{section:atb}

%Unlike the standard bargaining problem, we do not have a disagreement point, but only a bargaining set $C$. 
Given a bargaining set\footnote{Unlike standard bargaining models, in our case a bargaining problem contains a bargaining set only, without a disagreement point.} we define the \emph{random-extensive-form (REF) solution} to be 
\begin{align*}
\lim_{\varepsilon \rightarrow 0} \lim_{n \rightarrow \infty} \val (H^\varepsilon_n).
\end{align*}
Note that the REF solution can also be derived from a more natural version of a random extinctive form game $R_n$, where the REF solution is the limit point of the medians of the marginal distributions (see Proposition \ref{pro:second-part}).

We recall that a solution is called \emph{standard} if it satisfies Pareto efficiency, Symmetry, and Invariance with respect to positive affine transformations.
\begin{corollary}
The REF solution is a standard solution.
\end{corollary}
\begin{proof}
Pareto efficiency and symmetry follow immediately from Theorem \ref{theo:bar}. For invariance with respect to positive affine transformations we note that the uniform distribution remains uniform after operating such a transformation. In addition, the subgame perfect equilibria outcomes are invariant to such transformations.
\end{proof}

Beyond the fact that the REF solution satisfies several desirable axioms, the main advantage of the REF solution is the fact that this solution has a very clean probabilistic (approximate) implementation. We set small $\varepsilon$ and large $n$. 
%(latter we discuss who to set $\varepsilon$ and $n$). 
Now we simply randomize the decision makers and the outcomes in the game as it is done in the hybrid game $H^\varepsilon_n$. With high probability the subgame perfect equilibrium outcome of the game is close to the REF solution.
%i.i.d uniformly at random  Assume that we want the solution to be within a distance of $\delta$ from the REF solution with probability $1-\delta$. We first set $\varepsilon$ that guarantees that the concentration point of $H^\varepsilon_n$ will be $\frac{\delta}{2}$-close to the REF solution (this is possible by Proposition \ref{pro:second-part}). Now we set large enough $n$ that guarantees that $1-\delta$ of the mass is concentrated within a distance of $\frac{\delta}{2}$ from the concentration point.

\subsection{Proof of Theorem \ref{theo:bar}}\label{sec:bar-pr}
We start with a proof of the first part of Theorem \ref{theo:bar} (regarding the concentration around a Pareto efficient point). In fact, we derive a stronger statement which provides also a type of characterization for the concentration point.

We denote by $\chi_n^\varepsilon$ the probability distribution of $\val(H^n_\varepsilon)$ and for $i=1,2$ we let $F^{i,\varepsilon}_n$ be the CDF of the marginal distribution of $\nu_n^\varepsilon$ over the payoffs to Player $i$. We define $x^\varepsilon_n$ to be the unique number $x$ such that $F^{1,\varepsilon}_n(x)=b^\varepsilon$. Similarly, $y_n^\varepsilon$ is the number $y$ such that $F^{2,\varepsilon}(y)=b_\varepsilon$. 
\begin{proposition}\label{pro:first-part}
The sequences $x^\varepsilon_n$ and $y^\varepsilon_n$ have the following properties.
\begin{enumerate}
\item[(A1)] $\{x^\varepsilon_{2n}\}_n$ and $\{y^\varepsilon_{2n+1}\}_n$ are monotonically increasing in $n$.
\item[(A2)] $\{(x^\varepsilon_n,y^\varepsilon_n)\}_n$ is a converging sequence.
\item[(A3)] $(x^\varepsilon_n,y^\varepsilon_n)\underset{n\longrightarrow \infty}{\rightarrow} (x^\varepsilon,y^\varepsilon)$, namely, the limit of the sequence is the concentration point of $\val(H^n_\varepsilon)$.
\item[(A4)] $(x^\varepsilon,y^\varepsilon)$ lies on the Pareto efficient boundary of $C$.
\end{enumerate}
\end{proposition}
\begin{proof}
By applying identical considerations to those in the proof of Theorem \ref{th:gd}, and replacing $b$ by $b^\varepsilon$ we get exactly all these properties.
\end{proof}

Now we get to the proof of the second statement in Theorem \ref{theo:bar}. The proof is based on a connection between the sequences of measures $\val(R_n)$ and $\val(H^\varepsilon_n)$ (see Proposition \ref{pro:second-part}). In order to state this connection we need an additional convergence statement about the random structure games $\val(R_n)$.

%We present here several notations for the sequence of measures $\{\mu_n\}_n=\{\val(R_n)\}_n$. 
We denote by $F^1_n,F^2_n$ the CDFs of the marginal distribution of $\nu_n=\val(R_n)$ over the payoffs to Player $1,2$. We denote by $x_n$ and $y_n$ the medians of $F^1_n$ and $F^2_n$ (respectively).

%Before getting to the proof of Theorem \ref{theo:bar}, we introduce some properties of the sequence of measures $\{\mu_n\}_n=\{\val(R_n)\}$ for the case of a (symmetric) random structure. These properties will be useful in the proof of Theorem \ref{theo:bar}.
\begin{proposition}\label{pro:sym}
The sequences $x_n$ and $y_n$ are monotonically increasing. Moreover
the limit point $(x,y)=\lim_{n\rightarrow \infty} (x_n,y_n)$ lies on the Pareto efficient boundary of\footnote{Unlike other structures of random games discussed in this paper, in the (symmetric) random structure we do not have a concentration of the limit measure around the point $(x,y)$. Nevertheless, the proposition states that the limit point $(x,y)$ is well defined, and is Pareto efficient.} $C$.   
\end{proposition}

\begin{proof}
For $a\in \mathbb{R}$ we denote $C_1(a)=\{c\in c: c_1<a\}$. In order to prove that $\{x_n\}_n$ is monotonically increasing it is sufficient to show that $\nu_{n+1}(C_1(x_n))\leq \frac{1}{2}$. 

With probability $\frac{1}{2}$ Player 1 controls the node, and in such a case, the result (at the $(n+1)$th iteration) is in $C_1(x_n)$ iff both i.i.d.\ realizations (according to $\nu_n$) are in $C_1(x_n)$. The probability of such an event is $(\nu_n(C_1(x_n)))^2=\frac{1}{4}$.

With probability $\frac{1}{2}$ Player 2 controls the node, and in such a case, a necessary condition for the result (at the $(n+1)$th iteration) to be in $C_1(x_n)$ is that at least one realization is in $C_1(x_n)$. The probability of such an event is $1-(\nu_n(C_1(x_n)))^2=\frac{3}{4}$.

Summarizing, over the two cases we get $\nu_{n+1}(C_1(x_n))\leq \frac{1}{2}\cdot \frac{1}{4} + \frac{1}{2}\cdot \frac{3}{4} =\frac{1}{2}$.

Similar arguments show that the sequence $\{y_n\}_n$ is monotonically increasing.

Now we turn to the proof of statement (2). 

Assume by way of contradiction that $(x,y)$ lies in the interior of $C$, and then $\nu_{0}(C_{>>}(x,y))>0$ (see Definition \ref{def:c}). By equation \eqref{eq:im1} (which holds for the random controlling player case) we have that for every $n,$
\begin{align*}
\nu_{n+1}(C_{>>}(x,y))\geq \nu_{n}(C_{>>}(x,y))(\nu_{n}(C_{>>}(x,y))+F^1_{n}(x)+F^2_n(y)).
\end{align*}
Since $x_n\leq x$ (and $y_n\leq y$) we have $F^1_{n}(x)+F^2_n(y)\geq 1$. Hence,
\begin{align*}
\nu_{n+1}(C_{>>}(x,y))\geq \nu_{n}(C_{>>}(x,y))(\nu_{n}(C_{>>}(x,y))+1),
\end{align*}
which contradicts the fact that $\nu_{n}(C_{>>}(x,y))\leq 1$ for every $n$.

Assume by way of contradiction that $(x,y)$ lies out of the $C$. Then for some finite $n'$ $(x_{n'},y_{n'})$ is out of the set $C$. We denote $A_1=\{c\in C: c_1 \geq x_{n'}\}$ and $A_2=\{c\in C: c_2 \geq y_{n'}\}$. Note that $A_1 \cap A_2 =\emptyset$, and $\nu_{n'}(A_1)=\nu_{n'}(A_2)=\frac{1}{2}$, which implies that $\nu_{n'}(C_{<<}(x_{n'},y_{n'}))=0$. This contradicts the fact that the density of $\nu_n$ remains strictly positive everywhere in $C$ for every finite $n$.

Therefore $(x,y)$ must lie on the Pareto efficient boundary.
\end{proof}

The connection between the sequences of measures $\val(R_n)$ and $\val(H^\varepsilon_n)$ is given by following proposition, which states that the concentration point of $\val(H^\varepsilon_n)$ is close to the limit of medians of $\val(R_n)$ from Proposition \ref{pro:sym}.
\begin{proposition}\label{pro:second-part}
$\lim_{\varepsilon \rightarrow 0} (x^\varepsilon,y^\varepsilon) = (x,y)$.
\end{proposition}
We emphasize that $\val(H^\varepsilon_n)$ and $\val(R_n)$ behave very differently for large $n$ (the first one concentrates around a point, whereas the latter does not). Therefore it is somewhat surprising that we succeeded to connect the concentration point of $\val(H^\varepsilon_n)$ with some object that is derived from the sequence $\{\val(R_n)\}$.

\begin{proof}[Proof of Proposition \ref{pro:second-part}]
Let $\lambda=\lambda(C)$ be a global Lipschitz constant such that for every two points $(a_1,a_2)$ and $(b_1,b_2)$ on the Pareto efficient boundary of $C$ holds\footnote{Actually we do not have to assume existence of a \emph{global} Lipschitz constant for the Pareto efficient boundary. Existence of a \emph{local} Lipschitz constant around the point $(x,y)$ will suffice for the arguments that we present here in the proof.}
\begin{align}\label{eq:lip}
\frac{1}{\lambda} |a_1-b_1| \leq |a_2-b_2|\leq \lambda |a_1-b_1|. 
\end{align}

We set $\delta>0$, and we shall prove that there exists $\varepsilon(\delta)>0$ such that for every $\varepsilon<\varepsilon(\delta)$ holds $||(x^\varepsilon,y^\varepsilon)-(x,y)||_\infty \leq \delta$. The way we do it is by considering the sequences $\{(x_n,y_n)\}_n$ and $\{(x^\varepsilon_n,y^\varepsilon_n)\}_n$.
%$\{(x^\varepsilon_n,y^\varepsilon_n)\}$ converges to $(x^\varepsilon,y^\varepsilon),$ (i.e., the limit of the sequence $\{(x^\varepsilon_n,y^\varepsilon_n)\}$ is the concentration point of $\val(R^\varepsilon_n)$) and that

By Proposition \ref{pro:sym} $\lim_n (x_n,y_n)=(x,y)$, so there exists $n_0$ such that 
\begin{align}\label{eq:xn-x}
|x_{2n_0}-x|\leq \frac{\delta}{3\lambda} \text{ and } |y_{2n_0+1}-y|\leq \frac{\delta}{3\lambda}.
\end{align}
We emphasize that $n_0$ does not depend on $\varepsilon$ (because $n_0$ depends only on a symmetric random structure where $\varepsilon$ is not involved).

We denote by $x_n^{med,\varepsilon}$ and $y_n^{med,\varepsilon}$ the medians of $F^{1,\varepsilon}_n$ and $F^{2,\varepsilon}_n$. By Lemma \ref{lem:str-con} the measures $\chi^\varepsilon_{2n_0}$ and $\chi^\varepsilon_{2n_0+1}$ are strongly continuous with respect to $\varepsilon$. In particular, it follows that the medians of $\chi^\varepsilon$ converge to the medians of $\nu$. Namely, $\lim_{\varepsilon\rightarrow 0} x_{2n_0}^{med,\varepsilon}=x_{2n_0}$ and $\lim_{\varepsilon\rightarrow 0} y_{2n_0+1}^{med,\varepsilon}=y_{2n_0+1}$. Therefore, there exists $\varepsilon'$ such that for every $\varepsilon<\varepsilon'$ holds
\begin{align}\label{eq:xmed-xn}
|x_{2n_0}^{med,\varepsilon} - x_{2n_0}|\leq \frac{\delta}{3\lambda} \text{ and } |y_{2n_0+1}^{med,\varepsilon} - y_{2n_0+1}|\leq \frac{\delta}{3\lambda}.
\end{align}

Note that $F^{1,\varepsilon}_{2n_0}(x)$ and $F^{2,\varepsilon}_{2n_0+1}(y)$ are continuous with respect to $x$ and $y$ (respectively). Note also that $\lim_{\varepsilon \rightarrow 0} b^\varepsilon =\frac{1}{2}$. Therefore by the definition of $x^{\varepsilon}_n$ and $y^\varepsilon_n$ there exists $\varepsilon''$ such that for every $\varepsilon<\varepsilon''$ holds
\begin{align}\label{eq:xneps-xmed}
|x^\varepsilon_{2n_0} - x_{2n_0}^{med,\varepsilon}|\leq \frac{\delta}{3\lambda} \text{ and } |y^\varepsilon_{2n_0+1} - y_{2n_0+1}^{med,\varepsilon}|\leq \frac{\delta}{3\lambda}.
\end{align}

If we set $\varepsilon(\delta)=\min(\varepsilon',\varepsilon'')$ then inequalities \eqref{eq:xn-x}, \eqref{eq:xmed-xn}, and \eqref{eq:xneps-xmed} guarantee that 
\begin{align*}
||(x_{2n_0},y_{2n_0+1})-(x,y)||_\infty \leq \frac{\delta}{\lambda}.
\end{align*}
By properties (1) and (3) in Proposition \ref{pro:first-part} (at the beginning of the proof) we know that $(x^\varepsilon,y^\varepsilon)$ Pareto dominates the point $(x_{2n_0},y_{2n_0+1})$. In addition, by Proposition \ref{pro:sym} we know that $(x,y)$ is Pareto efficient. By the definition of $\lambda$ together with  the above three properties we get that $||(x^\varepsilon,y^\varepsilon)-(x,y)||_\infty \leq \delta$.
\end{proof}

\begin{lemma}\label{lem:str-con}
For every constant $n_0$ the measure $\chi(\varepsilon)=\val(H^\varepsilon_{n_0})$ is strongly continuous (in the total variation distance) with respect to $\varepsilon$ around $\varepsilon=0$.
\end{lemma}
\begin{proof}
The random variable $\val(H^\varepsilon_{n_0})$ can be obtained by the following procedure. 
\begin{enumerate}
\item We randomize $2^{n_0}$ i.i.d.\ random outcomes according to $\nu=\mu_0$.
\item For each node $v$ we choose Player $i=1,2$ to control node $v$ with probability $\frac{1}{2}$.
\item For each node $v$ we draw a Bernoulli random variable $e_v=0,1$ with $Pr(e_v=1)=2\varepsilon$.
\item For each odd-depth node $v$, if $e_v=1$ we set Player 1 to control the node (irrespective of who was chosen to control it at step (2)).
\item For each even-depth node $v$, if $e_v=1$ we set Player 2 to control the node (irrespective of who was chosen to control it at step (2)).
\item We apply backward induction on the resulted game.
\end{enumerate}
%Note that when $\varepsilon \rightarrow 0$ the probability of having a node $v$ where $e_v=1$ goes to 0 (more precisely, this probability is equal to 
Note that $Pr[e_v=0 \text{ for all } v]= (1-2\varepsilon)^{2^{n_0}-1}$. In such a case (of $e_v=0$ for all $v$) we have exactly the random variable $\val( R_{n_0})$. Therefore the total variation distance is bounded by $d_{TV}(\val(R_{n_0}),\val(H^\varepsilon_{n_0}))\leq 1-(1-2\varepsilon)^{2^{n_0}-1}$, i.e., $\lim_{\varepsilon \rightarrow 0} d_{TV}(\val(R_{n_0}),\val(R^\varepsilon_{n_0}))=0$.

\end{proof}

Now Theorem \ref{theo:bar} follows immediately from Propositions \ref{pro:first-part} and \ref{pro:second-part}.
\begin{proof}[Proof of Theorem \ref{theo:bar}]
The first part of the Theorem is proved in Proposition \ref{pro:first-part} items (3) and (4). Regarding the second part, if $C$ is symmetric then $x_n=y_n$ for every $n$, and therefore obviously $x=y$. So Proposition \ref{pro:second-part} completes the proof.
\end{proof}
\section{Discussion}\label{sec:discussion}
\textbf{1. General settings for which all the results hold.} By looking on the proofs of the results one can observe that the properties of the domain $C$ which were used in the proofs are as follows:
\begin{enumerate}
\item $C$ is a compact set with a non-empty interior.
\item For every not Pareto efficient point $(x,y)\in C$ the set $C_{\geq \geq}(x,y)$ has a non-empty interior.
\item The Pareto efficient boundary forms a connected path which has the Lipschitz property. This requirement is needed only for the analysis of the hybrid game.
\end{enumerate}
These are very minimalistic requirements which include in particular the cases where $C$ is a simply connected set with a concave Pareto efficient boundary (rather than convex).

Again by looking at the proofs of the results one can observe that the only property on the initial distribution $\mu_0$ that was used is the fact that $\mu_0$ is a non-atomic measure with strictly positive density over all $C$. The only change in the results (but not in the proofs) that is caused by changing the uniform distribution by a general distribution $\mu_0$ is Theorem \ref{th:zs}: The concentration point is obviously not the point $b\approx 0.62$ but the point $c$ for which $\mu_0([0,c])=b$.

\textbf{2. Games over ternary trees.} The paper focuses on the case where the game tree is a complete binary tree. The techniques derived in the paper can be applied also for the case where the game tree is ternary. We overview here the results for this case (without proofs). As we can see part of the results are similar but part are not:
\begin{itemize}
\item The value of the alternating game concentrates around a Pareto efficient point. In the case of a zero-sum game the concentration is around $b=0.68$, where $b$ is the unique solution of $1-(1-x^3)^3=x$ in the segment $(0,1)$. These results are similar to the case of binary trees.
\item The value of the random controlling player game converges to the distribution $\tau$ which assigns a probability of $\frac{1}{2}$ to both points $m_1,m_2\in C$, where $m_i$ is the best point in $C$ for Player $i$. This statement holds for both, the zero sum case and the case of a general two-dimensional set $C$. These results \emph{defer} from the binary case where the limit measure has positive density over the entire Pareto efficient boundary.
\item For small enough values of $\varepsilon$, the value of the hybrid game converges to $\tau^\varepsilon$, where $\tau^\varepsilon$ assigns a probability $\varepsilon$ close to $\frac{1}{2}$ to $m_1$ (from the previous bullet) and the remaining probability to $m_2$. In particular, the hybrid game fails to have the desirable property of concentration around a point. Namely, a different random structure than the hybrid game is required in order to implement a standard solution in ternary games.
\end{itemize}

\end{document}